\documentclass[a4paper,oneside,11pt]{article}

\usepackage{amsmath,amsfonts,amscd,amssymb}
\usepackage{longtable,geometry}
\usepackage[english]{babel}
\usepackage[utf8]{inputenc}
\usepackage[active]{srcltx}
\usepackage[T1]{fontenc}
\usepackage{graphicx}
\usepackage{pstricks}
\usepackage{bbm}

\usepackage{MnSymbol}
\usepackage{stmaryrd}
\usepackage{nicefrac}
\usepackage{calrsfs}

 \usepackage{rotating}

\usepackage{xcolor}
\usepackage{framed}

\colorlet{shadecolor}{blue!15}

\newtheorem{theorem}{Theorem}[section]

\newtheorem{corollary}[theorem]{Corollary}
\newtheorem{lemma}[theorem]{Lemma}
\newtheorem{proposition}[theorem]{Proposition}
\newtheorem{definition}[theorem]{Definition}

\newtheorem{remark}[theorem]{Remark}

\newenvironment{proof}[1][\relax]
  {\paragraph{Proof\ifx#1\relax\else~of #1\fi}}%
  {~\hfill$\square$\par\bigskip}

\newcommand{\ee}{\end{equation}}
\newcommand{\be}{\begin{equation}}

\title{Random currents expansion of the Ising model}

\author{Hugo Duminil-Copin}

\begin{document}

\newcommand{\lr}[1][]{\stackrel{#1}\longleftrightarrow}

\newcommand{\nlr}[1][]{\overset{#1}{\not\longleftrightarrow}}

\maketitle

\begin{abstract}
Critical behavior at an order/disorder phase transition has been a central object of interest in statistical physics. In the past century, techniques borrowed from many different fields of mathematics (Algebra, Combinatorics, Probability, Complex Analysis, Spectral Theory, etc) have contributed to a more and more elaborate description of the possible critical behaviors for a large variety of models. The Ising model is maybe one of the most striking success of this cross-fertilization, for this model of ferromagnetism is now very well understood both physically and mathematically. In this article, we review an approach, initiated in \cite{GriHurShe70,Aiz82} and based on the notion of random currents, enabling a deep study of the model. 
\end{abstract}

\vspace{-2mm}
\section{The Ising model}

\subsection{Ising model on a finite graph} In the Ising model, a magnetic material is described as a collection of small magnetic moments placed regularly on a lattice. The magnetic property of the material is assumed to be highly anisotropic, in the sense that the  magnetic moments can point only in two opposite directions (which are usually called $\pm1$). The aim of the model is to explain how interactions between neighboring magnetic moments can create (or not) a global magnetization of the material under the application of an exterior magnetic field. We will mostly consider the case of ferromagnetic interactions, in which the interaction between neighboring magnetic moments pushes them to align (or equivalently to be equal).

Formally, the model is defined on a finite set $\Lambda$ as follows. A {\em spin} variable $\sigma_x\in\{\pm1\}$ is attributed to each $x\in \Lambda$. The
 {\em spin configuration} $\sigma=(\sigma_x:x\in \Lambda)\in\{\pm1\}^{\Lambda}$ is given by the collection of all the spins.
Let $\mathcal E=\mathcal E(\Lambda)$ be the set of unordered pairs $\{x,y\}$ of elements in $\Lambda$ with $x\ne y$. Below, we denote an element of $\mathcal E$ by $xy$. For a family $(J_{xy})_{xy\in\mathcal E}$ of {\em coupling constants} $J_{xy}\ge0$ and $h\in\mathbb R$, introduce the energy of a spin configuration $\sigma$ defined by
$$\mathbf H_{\Lambda,h}(\sigma):=-\sum_{xy\in\mathcal E}J_{xy}\,\sigma_x\sigma_y-h\sum_{x\in \Lambda}\sigma_x.$$
For $\beta\ge0$ and $f:\{\pm1\}^{\Lambda}\longrightarrow \mathbb  R$, introduce
\be\label{eq:4}Z_{\Lambda,\beta,h}(f):=\sum_{\sigma\in\{\pm1\}^{\Lambda}}f(\sigma)\exp[-\beta\mathbf H_{\Lambda,h}(\sigma)].\ee
\begin{definition}The Ising measure $\langle \cdot\rangle_{\Lambda,\beta,h}$ with coupling constants $(J_{xy})$ on $\Lambda$ at inverse temperature $\beta\ge0$ and external field $h\in\mathbb R$ is defined by the formula
\be\label{eq:Gibbs}\langle f\rangle_{\Lambda,\beta,h}:=\frac{Z_{\Lambda,\beta,h}(f)}{Z_{\Lambda,\beta,h}(1)}\qquad\qquad\text{for every $f:\{\pm1\}^{\Lambda}\longrightarrow \mathbb  R$.}\ee
\end{definition}
One often defines the Ising model on a graph $G$ with vertex-set $\Lambda$. In this context, if the coupling constants are defined in such a way that $J_{xy}=1$ if $xy$ is an edge of $G$ and 0 otherwise,
we speak of the  {\em nearest-neighbor ferromagnetic} (n.n.f.) Ising model on $G$. 

\subsection{Ising model on an infinite graph}The Ising model efficiently describes the phase transition\footnote{Pierre Curie discovered a transition between the paramagnetic (i.e.~the ability of a material to gain a magnetization when immersed in a magnetic field) and the ferromagnetic (i.e.~the ability of a material to keep this magnetization when the magnetic field is removed) behaviors of more that twenty real-life materials in his thesis in 1895. He discovered Curie's law for paramagnets at the same time.}  at Curie's temperature between the paramagnetic and the ferromagnetic properties of a material. 
In order to witness the emergence of a phase transition, we consider the model on infinite sets. For simplicity, we will focus on the Ising model at inverse-temperature $\beta$ on
$$\Lambda=\mathbb  Z^d:=\big\{(x_1,\dots,x_d):x_i\in \mathbb  Z\text{ for all }1\le i\le d\big\}$$
and assume that the coupling constants $J_{xy}\ge0$ depend only on $x-y$. In such case, we speak of a model which is {\em ferromagnetic} and {\em invariant (under translations)}. 

One cannot directly define the Ising model on $\mathbb Z^d$ by the same formulae as in the previous paragraph since the energy would involve a divergent series. Hence, we are bound to define the measure as the limit of measures on finite sets. One possible procedure is the following. For $n\ge1$, let $\Lambda_n$ denote the set $\{-n,\dots,n\}^d$. Then, $\langle\cdot\rangle_{\Lambda_n,\beta,h}$ can be proved to converge weakly (as $n\rightarrow\infty$) to a probability measure $\langle\cdot\rangle_{\beta,h}$ on $\{\pm1\}^{\mathbb Z^d}$ equipped with the $\sigma$-algebra $\mathcal F$ generated by the random variables $\sigma\mapsto \sigma_x$ for every $x\in \mathbb Z^d$.

Once infinite-volume measures $\langle\cdot\rangle_{\beta,h}$ are defined, we may introduce an order parameter measuring the magnetization of the material and speak of a phase transition.

\begin{definition} The {\em spontaneous magnetization} $m^*(\beta)=m^*(\beta,d)$ is the limit as $h\searrow 0$ of $m(\beta,h):=\langle \sigma_0\rangle_{\beta,h}$. The {\em critical inverse-temperature} of the model is defined as
\be\beta_c=\beta_c(d):=\inf\{\beta\ge0:m^*(\beta)>0\}.\ee
\end{definition}
The quantity $1/\beta_c$ should be interpreted as Curie's temperature. It separates a regime without spontaneous magnetization ($m^*(\beta)=0$) corresponding to a paramagnet from a regime with spontaneous magnetization ($m^*(\beta)>0$) corresponding to a ferromagnet. Physicists and mathematicians are then interested in the behavior of the model near $\beta_c$ describing the transition between the two regimes. 

\begin{remark}The Ising model goes back to Lenz \cite{Len20} who suggested it to his PhD student Ising. Ising  \cite{Isi25} proved that $\beta_c=\infty$ for the n.n.f. Ising model on $\mathbb Z$. Ising also conjectured that $\beta_c(d)$ is always equal to infinity and that the model is therefore unable to predict the existence of Curie's temperature. Because of this unfortunate prediction, the model was abandoned for some time before Peierls \cite{Pei36} finally contradicted Ising by proving that $\beta_c(d)\in(0,\infty)$ for any $d\ge 2$. \end{remark}

\subsection{Partition function of the Ising model} 

The quantity $Z_{\Lambda,\beta,h}(1)$ is called the {\em partition function} of the model (from now on, we will denote it by $Z_{\Lambda,\beta,h}$). The partition function of the model is directly connected to the {\em free energy} defined as
\be\nonumber
f(\beta,h):=\lim_{n\rightarrow \infty}\frac{1}{|\Lambda_n|}\log Z_{\Lambda_n,\beta,h}.
\ee
Some properties of the model can be obtained via the free energy -- for instance
$
\langle \sigma_0\rangle_{\beta,h}=\frac1{\beta}\frac{\partial}{\partial h}f(\beta,h)
$ --
and the existence of a phase transition is directly related to singularities (in $\beta$ and $h$) of $f(\beta,h)$. 
For these reasons, trying to compute the free energy of the model has been a central question in statistical physics. 

In \cite{Ons44}, Onsager built on works of Kramers and Wannier \cite{KraWan41} to show that the free energy of the n.n.f. Ising model with zero magnetic field in $\mathbb Z^2$ is given by
\be\nonumber
\beta f(\beta,0)=\ln (2)+\frac1{8\pi^2}\int_0^{2\pi}\int_0^{2\pi}d\theta_1d\theta_2 \ \ln\Big(\cosh(\beta)^2-\sinh(\beta)\big(\cos(\theta_1)+\cos(\theta_2)\big)\Big)
\ee
from which one infers that $\beta_c=\tfrac12\ln(1+\sqrt 2)$. Onsager's computation of the free energy is based on the study of the eigenvalues of the so-called transfer matrices. The original strategy used by Onsager is based on the fact that the transfer matrix is the product of two matrices whose commutation relations generate a finite dimensional Lie algebra. Later on, Kaufman \cite{Kau49} gave a simpler solution using Clifford algebra and anti-commuting spinor (free-fermion) operators.
Onsager also announced that
\be\nonumber
m^*(\beta)=\Big(1-\sinh(\beta)^{-4}\Big)^{1/8}
\ee
without providing a rigorous proof of the statement. The result was mathematically proved by Yang in \cite{Yan52} using a limiting process of transfer matrix eigenvalues. Later on, Onsager explained that he did not publish the proof because he was unable to justify certain statements regarding Toeplitz determinants.

Since Onsager's original computation of the free energy, many new approaches were proposed to compute the free energy. Yang and Baxter provided an alternative strategy based on the Yang-Baxter equation by greatly generalizing some of the ideas related to the star-triangle transformation introduced in Onsager's solution. The free energy of the model was also mapped to several other models. Kasteleyn \cite{Kas63} related the partition function of the Ising model  to the one of dimers, thus enabling him to study the n.n.f. Ising model on planar graphs. Kac and Ward \cite{KacWar52} provided an approach, referred to as the combinatorial approach, which expresses the free energy in terms of families of signed loops. Schultz, Mattis and Lieb mapped the Ising model to fermionic systems in \cite{SchMatLie64}. To illustrate the variety of solutions, Baxter and Enting \cite{BaxEnt78} provided yet another solution in a paper entitled ``399$^{th}$ solution to the 2D Ising model", which is based solely on the star-triangle transformation. 

The previous list of solutions of the n.n.f. Ising model in 2D is impressive, and it is fair to say that the model has been a laboratory for new techniques related to exact computations of partition functions for statistical physics models. 
However, for more general interactions or simply for the n.n.f. model in higher dimension, the approach based on an exact computation of the free energy seems much more challenging since the model is not currently known to be exactly solvable. 

Physicists and mathematicians therefore turned their attention to alternative approaches to handle the model. They started by studying expansions of the partition function. The most famous ones are called the low and high temperature expansions\footnote{The low temperature expansion enabled Peierls to show that the critical inverse-temperature of the n.n.f. Ising model on $\mathbb Z^d$ is strictly smaller than $\infty$ for $d\ge2$.}. An expansion in terms of subgraphs of the original graph, called the random-cluster model, was found by Fortuin and Kasteleyn in \cite{ForKas72}. 
Several random-walk expansions were also introduced \cite{BryFroSpe82,Sym66}. Last but not least, the so-called random currents expansion was developed in \cite{GriHurShe70,Aiz82}. 

The strength of all these expansions is that they work for all graphs. They do not lead to an explicit computation of the partition function or the free energy, but they provide new insight and often highlight specific properties of the model. For these reasons, their applications go beyond the original goal of expanding the partition function, since they enable the physicists and the mathematicians to prove new properties of the model. 

 The goal of this proceeding is to survey the results obtained via the random currents expansion. 
The article is organized as follows. We start by discussing a few expansions of the Ising model. Then, we focus on the random currents expansion and the fundamental ``switching lemma". In the fourth section we discuss some applications of random currents. Finally, the last section lists a few open questions related to random currents and the Ising model.

\section{Expansions of the Ising model correlations}

Fix $\beta,h\ge0$ and $A\subset \Lambda$ both finite. We set $\sigma_A=\prod_{x\in A}\sigma_x$. The goal of this section is to expand $Z_{\Lambda,\beta,h}(\sigma_A)$ as a sum of weighted objects and to deduce an expression for $\langle\sigma_A\rangle_{\Lambda,\beta,h}$ in terms of these objects. The objects under consideration will be either graphs, walks or integer-valued functions. In the two first cases, we speak of graphical and random-walk representations 
while in the last case, we speak of expansions in currents.

Before starting, let us make a small detour. The magnetic field $h\ge0$ can be seen as a global bias pushing spins towards $+1$. This magnetic field can be interpreted in a nice way  by introducing an addition point $\mathfrak g\notin \Lambda$ called {\em Griffiths' ghost vertex} and by setting $J_{x\mathfrak g}=J_{x\mathfrak g}(h):=h$. Then, 
$\langle \cdot\rangle_{\Lambda,\beta,h}=\langle\cdot_{|\Lambda} \,|\,\sigma_{\mathfrak g}=+1\rangle_{\Lambda\cup\{\mathfrak g\},\beta,0}$.
In other words, by adding one ``ghost" vertex, the magnetic field can be interpreted as an Ising model without magnetic field, conditioned on the spin of the ghost vertex to be $+1$. We will often use this interpretation with the ghost vertex for which we set $\sigma_{\mathfrak g}$ to be always $+1$ for obvious reasons.
From now on, $\mathcal E=\mathcal E(\Lambda\cup\{\mathfrak g\})$ is the set of unordered pairs $\{x,y\}\subset \Lambda\cup\{\mathfrak g\}$. 

\subsection{Expansion in integer-valued functions (or currents)}\label{sec:current} While the expansion in integer-valued functions is not the oldest nor the most elementary one, it is the one that will be studied in detail later in the text, and we therefore choose to present it first for full awareness. Let $\mathbb N:=\{0,1,2,\dots\}$ be the set of non-negative integers.
For $\mathbf n=(\mathbf n_{xy}:xy\in\mathcal E)\in \mathbb N^{\mathcal E}$ and $x\in \Lambda\cup\{\mathfrak g\}$, introduce $X(\mathbf n,x):=\sum_{y\in \Lambda\cup\{\mathfrak g\}}\mathbf n_{xy}$. We also set 
$\partial \mathbf n:=\{x\in \Lambda:X(\mathbf n,x)\text{ odd}\}
$ (note that $\mathfrak g$ is {\em never} in $\partial\mathbf n$).
As observed by Griffiths, Hurst and Sherman~\cite{GriHurShe70}, the identity
\be\label{eq:41}\exp[\beta J_{xy} \sigma_x\sigma_y]=\sum_{{\mathbf n}_{xy}=0}^\infty\frac{(\beta J_{xy}\sigma_x\sigma_y)^{{\mathbf n}_{xy}} }{{\mathbf n}_{xy}!}\ee
allows us to write
\be\label{eq:42}
Z_{\Lambda,\beta,h}(\sigma_A)=\sum_{\sigma\in\{\pm1\}^\Lambda}\sigma_A\sum_{{\mathbf n}\in \mathbb N^{\mathcal E}}\prod_{xy\in \mathcal E}\frac{(\beta J_{xy}\sigma_x\sigma_y)^{{\mathbf n}_{xy}} }{{\mathbf n}_{xy}!}.
\ee
By interchanging the two sums on the right-hand side, \eqref{eq:42} can be rewritten as
\be\label{eq:12}
Z_{\Lambda,\beta,h}(\sigma_A)=\sum_{{\mathbf n}\in \mathbb N^{\mathcal E}}w(\mathbf n)\sum_{\sigma\in\{\pm1\}^\Lambda}\prod_{x\in \Lambda}\sigma_x^{\,\mathbb{I}[x\in A]+X(\mathbf n,x)},
\ee
where 
\be\nonumber
w(\mathbf n)=w_{\Lambda,\beta,h}(\mathbf n):=\prod_{xy\in \mathcal E}\frac{\displaystyle(\beta J_{xy})^{{\mathbf n}_{xy}}}{{\mathbf n}_{xy}!}.
\ee
The trick comes next.
Fix $x\in \Lambda$ and $\sigma\in\{\pm1\}^\Lambda$. Define the configuration $\sigma^{(x)}$ obtained from $\sigma$ by reversing the spin at $x$. Since for a fixed $x\in\Lambda$, the map $\sigma\mapsto \sigma^{(x)}$ is an involution, and since the contributions of $\sigma^{(x)}$ and $\sigma$ to the sum over spin configurations in \eqref{eq:12} cancel each others as soon as $\mathbb{I}[x\in A]+X(\mathbf n,x)$ is odd, we find that
\be\label{eq:pm1trick}\sum_{\sigma\in\{\pm1\}^{\Lambda}}\prod_{x\in \Lambda}\sigma_x^{\,\mathbb{I}[x\in A]+X(\mathbf n,x)}=\begin{cases}\ 2^{|\Lambda|} &\text{ if $\partial \mathbf n=A$,}\\ \ \ 0&\text{ otherwise.}\end{cases}\ee
In conclusion,

 \be \label{eq:expansion current}
Z_{\Lambda,\beta,h}(\sigma_A)=2^{|\Lambda|}\sum_{{\mathbf n}\in\mathbb N^{\mathcal E}:\,\partial{\mathbf n}=A}w({\mathbf n})
.\ee
Since 
$
\langle\sigma_A\rangle_{\Lambda,\beta,h}=Z_{\Lambda,\beta,h}(\sigma_A)/Z_{\Lambda,\beta,h},$
we deduce that 
\be\label{eq:77}
\langle\sigma_A\rangle_{\Lambda,\beta,h}=\frac{\sum_{\mathbf n\in\mathbb N^{\mathcal E}:\partial \mathbf n=A}w(
\mathbf n)}{\sum_{\mathbf n\in\mathbb N^{\mathcal E}:\partial \mathbf n=\emptyset}w(
\mathbf n)}.\ee

\begin{remark}Equation \eqref{eq:77} implies the Griffiths' first inequality \cite{Gri67}: for any $\beta,h\ge0$ and $A\subset \Lambda$,
$\langle\sigma_A\rangle_{\Lambda,\beta,h}\ge0$.
\end{remark}

We conclude this section by introducing some terminology. An element of $\mathbb N^{\mathcal E}$ is called a {\em current}, and a vertex $x$ with $X(\mathbf n,x)$ odd a {\em source} of the current $\mathbf n$.
The name of current comes from the interpretation of the so-called backbone (see Section~2.3.2 below) of $\mathbf n$ as currents going from one source to another one.

\subsection{Graphical expansions of the Ising model} We now present three graphical representations, i.e. expansions in subsets of $\mathcal E$.
\medbreak
\noindent{\em 2.2.1 The high-temperature expansion.} For $E\subset \mathcal E$ and $x\in \Lambda$, set $\Delta(E,x)$ to be the number of pairs in $E$ containing $x$ and 
$\partial E:=\{x\in \Lambda:\Delta(E,x)\text{ odd}\}.
$
As observed by van der Waerden~\cite{vdW}, the identity
$$
\exp(\beta J_{xy}\sigma_x\sigma_y)=\cosh(\beta J_{xy})(1+\tanh(\beta J_{xy})\sigma_x\sigma_y)
$$
allows us to write 
\be
Z_{\Lambda,\beta,h}(\sigma_A)=c_0\ \sum_{\sigma\in\{\pm1\}^\Lambda}\sigma_A\sum_{E\subset \mathcal E}\prod_{xy\in E}\tanh(\beta J_{xy})\sigma_x\sigma_y,\ee
where 
$c_0=c_0(\Lambda,J):=\prod_{xy\in\mathcal E}\cosh(\beta J_{xy}).
$
Interchanging the two sums, we obtain
\be\nonumber
Z_{\Lambda,\beta,h}(\sigma_A)=c_0\sum_{E\subset \mathcal E}x(E)\sum_{\sigma\in\{\pm1\}^\Lambda}\prod_{x\in \Lambda}\sigma_x^{\,\mathbb{I}[x\in A]+\Delta(E,x)},
\ee
where
$x(E)=x_{\Lambda,\beta,h}(E):=\prod_{xy\in E}\tanh(\beta J_{xy}).$
The same use of the $\pm 1$ symmetry as in the previous section implies that \eqref{eq:pm1trick} is true with $E$ and $\Delta$ replacing $\mathbf n$ and $X$  so that 
\begin{align}\label{eq:high temperature expansion}
Z_{\Lambda,\beta,h}(\sigma_A)&=c_0\,2^{|\Lambda|}\sum_{E\subset\mathcal E:\partial E=A}x(E)
\end{align}
and
\be\label{eq:78}
\langle\sigma_A\rangle_{\Lambda,\beta,h}=\frac{\sum_{E\subset\mathcal E:\partial E=A}x(E)}{\sum_{E\subset\mathcal E:\partial E=\emptyset}x(E)}.\ee
This expansion is called the {\em high-temperature expansion} of the Ising model. \bigbreak
\noindent{\em 2.2.2 The Fortuin-Kasteleyn expansion.} Let us start with the partition function (i.e. $A=\emptyset$). If $p_{xy}:=1-\exp[-2\beta J_{xy}]$ for all $xy\in\mathcal E$, we may use the identity 
$\exp[\beta J_{xy}\sigma_x\sigma_y]=\exp[\beta J_{xy}](p_{xy}\mathbb I[\sigma_x=\sigma_y]+1-p_{xy})$
to get (after expanding)
\begin{align*}
Z_{\Lambda,\beta,h}
&=c_1\sum_{\sigma\in\{\pm 1\}^{\Lambda}}\sum_{E\subset \mathcal E}\Big(\prod_{xy\in E}p_{xy}\mathbb I[\sigma_x=\sigma_y]\Big)\Big(\prod_{xy\notin E}(1-p_{xy})\Big),\end{align*}
where $c_1=c_1(\Lambda,J):=\prod_{xy\in \mathcal E}\exp[\beta J_{xy}]$. Let $E(\sigma)$ be the set of $xy$ with $\sigma_x=\sigma_y$. Then, one may check that $\prod_{xy\in E}\mathbb I[\sigma_x=\sigma_y]=\mathbb I[E\subset E(\sigma)]$ so that
 \begin{align*}
Z_{\Lambda,\beta,h}
&=c_1\sum_{\sigma\in\{\pm 1\}^{\Lambda}}\sum_{E\subset E(\sigma)} \Big(\prod_{xy\in E}p_{xy}\Big)\Big(\prod_{xy\notin E}(1-p_{xy})\Big)
\nonumber,\end{align*}
where in the first line we used that the term in parentheses equals 1. Interchanging the two sums, we obtain
\begin{align*}
Z_{\Lambda,\beta,h}&=c_1\sum_{E\subset \mathcal E} \#\{\sigma\in\{\pm 1\}^{\Lambda}\text{ such that }E(\sigma)\supset E\}\cdot \Big(\prod_{xy\in E}p_{xy}\Big)\Big(\prod_{xy\notin E}(1-p_{xy})\Big).
\end{align*}
Now, the number of configurations $\sigma$ with $E(\sigma)\supset E$ is equal to $2^{k(E)}$, where $k(E)$ is the number of connected components of the graph $G(E)$ with vertex-set $\Lambda\cup\{\mathfrak g\}$ and edge-set $E$ (simply observe that the condition is that $\sigma$ must be constant on each connected component of $E$).
In conclusion,
$
Z_{\Lambda,\beta,h}=c_1 \sum_{E\subset \mathcal E} r(E)
$
with 
\be\nonumber
r(E)=r_{\Lambda,\beta,h}(E):=2^{k(E)}\cdot \Big(\prod_{xy\in E}p_{xy}\Big)\Big(\prod_{xy\notin E}(1-p_{xy})\Big).
\ee
More generally, one may easily check that 
\be\nonumber
Z_{\Lambda,\beta,h}(\sigma_A)=c_1 \sum_{E\in \mathfrak F_A} r(E),
\ee
where $\mathfrak F_A$ is the set of $E\subset\mathcal E$ such that each connected component of the graph $G(E)$ intersects $A$ (resp. $A\cup\{\mathfrak g\}$) an even number of times -- which can be zero -- if $|A|$ is even (resp. $|A|$ is odd). For instance, if $A=\{x,y\}$, then $E\in\mathfrak F_A$ if and only if $x$ and $y$ are in the same connected component of $E$.

While the Fortuin-Kasteleyn expansion is a graphical expansion exactly like the high-temperature expansion, there is no restriction on the possible sets $E$. This motivates the introduction of the probability measure $\phi_{\Lambda,\beta,h}$ on subsets of $\mathcal E$ such that 
\be\nonumber
\phi_{\Lambda,\beta,h}(\{E_0\})=\frac{r(E_0)}{ \sum_{E\subset\mathcal E} r(E)}
\ee
for any $E_0\subset\mathcal E$. This measure, introduced by Fortuin and Kasteleyn in \cite{ForKas72}, is now called the random-cluster model on $\Lambda\cup\{\mathfrak g\}$. With this notation, we can write
\be\label{eq:79}
\langle\sigma_A\rangle_{\Lambda,\beta,h}=\phi_{\Lambda,\beta,h}(\mathfrak F_A).
\ee
This expression of $\langle\sigma_A\rangle_{\Lambda,\beta,h}$ is of a very different nature than \eqref{eq:77} and \eqref{eq:78}. Indeed, the quantity is expressed as the probability of an event under a certain probability measure. This is not the case of the other expressions which involve different objects in the numerator and the denominator, and therefore cannot be directly interprated as a probability.

\bigbreak
\noindent{\em 2.2.3 The low-temperature expansion.} The simplest and oldest expansion is given by the low-temperature expansion. It is based on the observation that spin configurations are in correspondence with subsets of $\mathcal E$. For $\sigma\in\{\pm 1\}^{\Lambda}$, set $C(\sigma):=\{xy\in\mathcal E,\sigma_x\ne\sigma_y\}$. This ``contour set'' is the complement of $E(\sigma)$ (defined in Section 2.2.2) in $\mathcal E$. Note that for $h\ne 0$ (resp. $h=0$), $C:\sigma\longmapsto C(\sigma)$ is a one-to-one (resp. two-to-one) map. Let $\mathcal E^*$ be the image of the map $C$. Then,
\begin{align*}
Z_{\Lambda,\beta,h}&=c_1\sum_{\sigma\in\{\pm 1\}^{\Lambda}}\prod_{xy\in C(\sigma)}\exp[-2\beta J_{xy}]=c_2\sum_{E\subset \mathcal E^*}t(E),\end{align*}
where $t(E)=\prod_{xy\in E}\exp[-2\beta J_{xy}]$ and  $c_2$ is equal to $c_1$ if $h\ne0$ and $2c_1$ if $h=0$. We do not wish to spend time studying this expansion, but let us mention that $\langle\sigma_A\rangle_{\Lambda,\beta,h}$ has also a nice expression in this setting\footnote{We leave this problem as an exercise for the reader, and recommend that the reader starts with the case $h\ne 0$ and $A=\{x\}$.}.
The low-temperature expansion was used by Peierls \cite{Pei36} to study the regime $\beta\gg1$ and prove that  $\beta_c(2)<\infty$. In the case of a n.n.f. model on a planar graph, Kramers and Wannier \cite{KraWan41} related the low-temperature expansion to the high-temperature expansion on the dual graph (this relation is now known as the Kramers-Wannier duality).

\subsection{Random-walk expansions}$ $\\
\noindent{\em 2.3.1 Symanzik-Brydges-Fr\"ohlich-Spencer expansion.}
In \cite{Sym66}, Symanzik proposed an expansion of $Z_{\Lambda,\beta,h}(\sigma_A)$ in terms of random-walks, i.e.~finite sequences of elements in $\Lambda\cup\{\mathfrak g\}$. The idea is to write 
\be\nonumber
Z_{\Lambda,\beta,h}(\sigma_A)=\sum_{(\omega_1,\dots,\omega_p)\in \Omega(A)}Z(\omega_1,\dots,\omega_p),
\ee
where $\Omega(A)$ is the set of families $(\omega_1,\dots,\omega_p)$ of walks satisfying that $p=\lceil |A|/2\rceil$ and the vertices in $A$ if $|A|$ is even (resp. $A\cup\{\mathfrak g\}$ if $A$ is odd) are either the beginning or the end of one of the walks. We also allow the family to be empty if $A=\emptyset$, so that we can write (for some function $Z(\cdot)$)
\be\label{eq:80}
\langle\sigma_A\rangle_{\Lambda,\beta,h}=\displaystyle\sum_{(\omega_1,\dots,\omega_p)\in\Omega(A)}\frac{Z(\omega_1,\dots,\omega_p)}{Z(\emptyset)}.\ee
In \cite{Sym66}, Symanzik obtained such an expansion for $\phi^4$ models.
Brydges, Fr\"ohlich and Spencer adapted Symanzik's expansion to the discrete setting in \cite{BryFroSpe82}. The proof is based on an expansion in Taylor series in Fourier space (which looks like some type of high-temperature expansion), followed by a clever partial resummation of the terms. It is based on a family of integration-by-part formulae. 

We do not provide additional details here since we will be focusing on a closely related random-walk expansion, called the {\em backbone expansion}. The applications of these two expansions are roughly the same, but we choose to focus on the later one due to its direct connection with the expansion in currents.
\medbreak
\noindent{\em 2.3.1 The backbone expansion.} Let us start with the following connection between integer-valued functions and occupation times of a family of walks. Below, $n(\omega,x,y)$ denotes the number of times $t$ at which $\omega(t)=x$ and $\omega(t+1)=y$ (i.e.~the number of times the walk uses the oriented edge from $x$ to $y$).

\begin{proposition}
Fix $\mathbf n=\mathbb N^{\mathcal E}$ and a family of walks $\{\omega_1,\dots,\omega_k\}$ such that 
\be\nonumber
\mathbf n_{xy}=\sum_{i=1}^k \big(n(\omega_i,x,y)+n(\omega_i,y,x)\big).
\ee
Then $\partial\mathbf n=A$ if and only if there exists $\{\omega_{i_1},\dots,\omega_{i_p}\}\subset \{\omega_1,\dots,\omega_k\}$ such that the beginnings and ends of the walks $\omega_{i_1},\dots,\omega_{i_p}$ partition $A$, and the other walks are all loops (i.e.~that they start and end at the same place).
\end{proposition}

\begin{proof}
One direction is very simple to check. If the family of walks satisfies this property, then $\partial \mathbf n=A$. The other direction is not much more difficult. Fix $\mathbf n\in \mathbb N^{\mathcal E}$. We introduce a peeling procedure as follows. Order the elements of $\Lambda\cup\{\mathfrak g\}$ and $\mathcal E$ (the choice of the orderings is not important). Then, construct $x(t)$ and $\mathbf n(t)$ inductively as follows: set $x(0)$ to be the smallest (according to the fixed ordering) element of $A$ and $\mathbf n(0)=\mathbf n$ and for any $t\ge0$,
\medbreak
\noindent {\em If} $\mathbf n(t)=0$, then stop the induction.
 
 \noindent {\em Else if} $\mathbf n(t)_{x(t)y}=0$ for any $y\in\Lambda\cup\{\mathfrak g\}$, then set $\mathbf n(t+1)=\mathbf n(t)$ and

- {\em If} $A\subset\{x(0),\dots,x(t)\}$, then set $x(t+1)$ be such that $X(\mathbf n(t),x(t+1))>0$.

- {\em Else} let $x(t+1)$ be the smallest $x\in A$ not yet visited.

\noindent {\em Else} set $x(t+1)$ to be the smallest vertex for which $\mathbf n(t)_{x(t)x(t+1)}>0$ and 
$$\mathbf n(t+1)_{xy}=\begin{cases}\mathbf n(t)_{xy}-1&\text{ if }xy=x(t)x(t+1),\\
\mathbf n(t)_{xy}&\text{ otherwise}.\end{cases}$$

In words, we walk along the edges with positive current and subtract one to the current at each step that we make. When it is not possible to walk (because $X(\mathbf n(t),x(t))$ is zero at the vertex $x(t)$), we jump to a new vertex with $X(\mathbf n(t),x)>0$. The orderings help us  in case of multiple choices.  Now, if $t_1,\dots, t_k$ denote the times at which $\mathbf n(t_i)=\mathbf n(t_i-1)$, we obtain the family $\{\omega_1,\dots,\omega_k\}$ by setting $\omega_i(s)=\omega_i(t_i+s)$ for $s\le t_{i+1}-t_i$ and $i\le k$.
\end{proof}
The family $\{\omega_1,\dots,\omega_k\}$ is not unique. Nevertheless, if the orderings on vertices and edges is fixed, the procedure described in the above proof provides us with a family of walks. The family $(\omega_1,\dots,\omega_{\lceil |A|/2\rceil})$ is called the {\em backbone} of $\mathbf n$ (the other walks in the construction are loops). Then, \eqref{eq:expansion current} immediately implies \eqref{eq:80} with \be\nonumber
Z(\omega_1,\dots,\omega_p):=\sum_{\mathbf n\text{ with backbone }\omega_1,\dots,\omega_p}w(\mathbf n).\ee 

\section{The switching lemma and the random currents representation}

From now on we focus on the expansion in terms of currents. One of the main goals is to rewrite correlations functions in terms of probability of events for a probability measure on currents\footnote{Indeed, we saw that \eqref{eq:77}, for instance, is very different from \eqref{eq:79}, since it involves different currents in the numerator and denominator, and therefore cannot be interpreted as a probability.}. The following perspective on the Ising model's phase transition  is driven by the observation that the onset of long range order coincides with a percolation phase transition for currents.   This point of  view
was developed in \cite{Aiz82} and a number of subsequent works. 

From now on, summations involving the variable $\mathbf n$ will always be assumed to be summations over currents. For this reason, we drop $\mathbf n\in \mathbb N^{\mathcal E}$ from the notation. Let $A\Delta B$ denote the symmetric difference $(A\setminus B)\cup(B\setminus A)$. 
Also, associate  to a current $\mathbf n$ the subset $\widehat {\mathbf n}:=\{xy\in\mathcal E:\mathbf n_{xy}\ge1\}\subset \mathcal E$. Recall the definition of $\mathfrak F_A$ from the previous section.

The following statement, called the {\em switching lemma}, provides a powerful tool for studying currents. It was introduced in \cite{GriHurShe70} and then used at its full extent and popularized in \cite{Aiz82}. The proof is fairly simple and purely combinatorial. 
\begin{lemma}[Switching lemma]\label{switching}
For any $A,B\subset \Lambda$  and any function
 $F:\mathbb N^{\mathcal E}\rightarrow \mathbb C$: 
\be\nonumber
\sum_{\substack{\partial {\mathbf n}_1=A\\ \partial {\mathbf n}_2=B}}F({\mathbf n}_1+{\mathbf n} _2)w({\mathbf n}_1)w({\mathbf n}_2)=\sum_{\substack{\partial {\mathbf n}_1=\emptyset\\ \partial {\mathbf n}_2=A\Delta B}}F({\mathbf n}_1+{\mathbf n} _2)w({\mathbf n}_1)w({\mathbf n}_2){\mathbb  I}[\widehat{{\mathbf n}_1+{\mathbf n} _2}\in\mathfrak F_A].\ee
\end{lemma}

This lemma is particularly useful when considering sums of two currents\footnote{The idea of duplicating the system, meaning taking two independent Ising models, was already used by Griffiths \cite{Gri67} to prove his famous inequalities. The same strategy was later used by Lebowitz \cite{Leb74} for instance, who attributes the idea to Percus.}  since it enables to switch sources from one current to the other. For instance, applying the switching lemma to the square of \eqref{eq:77} (choosing $A=B$ and $F=1$) leads to
\be\label{eq:aaab}
\langle \sigma_A\rangle_{\Lambda,\beta,h}^2=\frac{\displaystyle\sum_{\partial {\mathbf n}_1=\partial {\mathbf n}_2=\emptyset}w({\mathbf n}_1)w({\mathbf n}_2){\mathbb  I}[\widehat{{\mathbf n}_1+{\mathbf n} _2}\in\mathfrak F_A]}{\displaystyle\sum_{\partial {\mathbf n}_1=\partial {\mathbf n}_2=\emptyset}w({\mathbf n}_1)w({\mathbf n}_2)}.
\ee
While taking the square of the correlation functions can appear as a big sacrifice, notice that the sums on the numerator and denominator are now on the same objects (namely pairs of sourceless currents). This seems to fit in the framework of probability theory, and it therefore calls for the following definition.
 \begin{definition}[Distribution on currents]\label{def:probability}
Fix $\beta,h\ge0$ and $A\subset \Lambda$. Define the distribution $\mathbf P^A=\mathbf P_{\Lambda,\beta,h}^A$ on currents such that for any $\mathbf n_0$ with $\partial\mathbf n_0=A$,
\be\mathbf P^A[\{\mathbf n_0\}]:=\frac{w({\mathbf n_0})}{\sum_{\partial\mathbf n=A}w(\mathbf n)}.\ee
\end{definition}

Note that ${\bf P}^A$ is supported on $\{\mathbf n:\partial\mathbf n=A\}$. Let $\nu \otimes \mu$ denote the product of the measures $\nu$ and $\mu$. With this notation, we deduce from \eqref{eq:aaab} that 
\begin{proposition}
For any $\beta,h\ge0$ and any $A\subset\Lambda$,
\be\label{eq:spin to RC}
\langle \sigma_A\rangle_{\Lambda,\beta,h}^2:={\bf P}^\emptyset\otimes{\bf P}^\emptyset\big[\widehat{{\mathbf n}_1+{\mathbf n}_2}\in \mathfrak F_A\big].
\ee
\end{proposition}
One cannot miss the comparison with \eqref{eq:79}: both right-hand sides involve the probability of the event $\mathfrak F_A$ for two different laws on random subsets of $\mathcal E$. It is important to note that in one case the left-hand side is the spin-spin correlation, while in the second case it is its square.

Similar applications of the switching lemma imply more general statements. For instance, Griffiths' second inequality 
$$
\langle \sigma_A\sigma_B\rangle_{\Lambda,\beta,h}\ge \langle\sigma_A\rangle_{\Lambda,\beta,h}\langle\sigma_B\rangle_{\Lambda,\beta,h}$$ (see \cite{Gri67})
is obtained by observing that
\be\nonumber
1-\frac{\langle\sigma_A\rangle_{\Lambda,\beta,h}\langle\sigma_B\rangle_{\Lambda,\beta,h}}{\langle \sigma_A\sigma_B\rangle_{\Lambda,\beta,h}}={\bf P}^\emptyset\otimes{\bf P}^{A\Delta B}\big[\widehat{{\mathbf n}_1+{\mathbf n}_2}\notin \mathfrak F_A\big]\ge0.
\ee
Note that one gets immediately that spin-spin expectations are increasing in $\beta\ge0$ (and one may prove the same in $h\ge0$) since
\be\nonumber
\frac{d}{d\beta}\langle\sigma_A\rangle_{\Lambda,\beta,h}=\sum_{xy\in \mathcal E}J_{xy}\big(\langle\sigma_A\sigma_x\sigma_y\rangle_{\Lambda,\beta,h}-\langle\sigma_A\rangle_{\Lambda,\beta,h}
\langle\sigma_x\sigma_y\rangle_{\Lambda,\beta,h}\big)\ge0.
\ee
\begin{remark}\label{rmk:comparison}
The random set $\widehat{\mathbf{n}}\subset\mathcal E$ with law ${\bf P}^\emptyset_{\Lambda,\beta,h}$ can be directly related to the high-temperature expansion and the random-cluster model. Indeed, consider a random variable $E\subset\mathcal E$ with law $\mu^\emptyset_{\Lambda,\beta,h}$ attributing probability proportional to $x(E)$ to each $E$ with $\partial E=\emptyset$ and 0 otherwise. Adding to $E$ each $xy\in \mathcal E$ independently with probability $1-1/\cosh(\beta J_{xy})$ gives a random variable with law ${\bf P}_{\Lambda,\beta,h}^\emptyset$ (this is fairly easy to see by noticing that the set $E$ plays the same role as the set of $xy\in\mathcal E$ with $\mathbf n_{xy}$ odd). Adding to this new random graph each $xy\in\mathcal E$ independently with probability $1-\exp(-\beta J_{x,y})$ leads to a random subset of $\mathcal E$ with law $\phi_{\Lambda,\beta,h}$ (see \cite{GriJan09,LupWer15}). In words, the configuration of random currents is sandwiched between the high-temperature and the random-cluster configurations.
\end{remark}

\section{Three applications of random currents}

The strength of the random currents representation is the alliance of two possible points of view: first, the backbone of a current can be interpreted as a family of walks, and second, the trace of currents (which is a subset of $\mathcal E$) can be used to express correlations in the model. In words, the currents provide both a random-walk and a percolation interpretation for the Ising model. In the following sections, we describe three applications of random currents. Each one of them relies directly or indirectly on properties inspired by both points of view.

For simplicity, we restrict our attention to the n.n.f. Ising model on $\mathbb Z^d$.

\subsection{Sharpness of the phase transition}

The critical parameter $\beta_c$ discriminates between an ordered regime ($m^*(\beta)>0$) and a disordered regime ($m^*(\beta)=0$). It is not difficult to see that 
$\langle\sigma_x\sigma_y\rangle_{\beta,0}$ remains bounded away from zero (respectively tends to zero) in the ordered (respectively disordered) regime. One is naturally led to the question of the speed of decay to zero when $\beta<\beta_c$. 

In 1987, Aizenman, Barsky and Fernandez \cite{AizBarFer87} used random currents to prove that the speed of decay is exponential (Property (3) of Theorem~\ref{thm:1} below). As a byproduct of their proof, they also showed that the magnetization satisfies mean-field lower bounds when $h=0$ and $\beta\searrow \beta_c$, and when $\beta=\beta_c$ and $h\searrow 0$ (Properties (1) and (2) of Theorem~\ref{thm:1}). We present a simplified version of the results here (the constants are not optimized). 

\begin{theorem}\label{thm:1} Consider the n.n.f. Ising model on $\mathbb Z^d$, then

(1) There exists $c_0\in (0,\infty)$ such that for $\beta>\beta_c$, $m^*(\beta)\ge c_0(\beta-\beta_c)^{1/2}$.

(2) There exists $c_1\in(0,\infty)$ such that for any $h\ge0$, $m(\beta_c,h)\ge c_1h^{1/3}$.

(3) For $\beta<\beta_c$, there exists $c_2=c_2(\beta)>0$ such that $$\langle\sigma_0\sigma_x\rangle_{\beta,0}\le \exp[-c_2 \|x\|]$$

$\ \ \ \ $ for all 
 $x\in \mathbb Z^d$, where $\|\cdot\|$ is the $\ell^1$-norm on $\mathbb R^d$.
\end{theorem}

An alternative proof also relying on random currents was provided recently in \cite{DumTas15}. Let us summarize it now. The proof of \cite{DumTas15} is based on the following quantity:
 For $\beta\ge0$ and a finite subset $S$ of $\mathbb Z^d$, set 
$\partial S:=\{x\in S:\exists y\notin S\text{ neighbor of }x\text{ in }\mathbb Z^d\}$
 and define
\be\nonumber
\varphi_S(\beta):=\sum_{x\in \partial S}\langle \sigma_0\sigma_x\rangle_{S,\beta,0}.
\ee
 This quantity is related to the magnetization through the following inequality
\be\label{eq:diff}
\frac{\partial}{\partial\beta}\big(m(\beta,h)^2\big)\ge c_3\,\Big(\inf_{S\ni 0}\varphi_\beta(S)\Big)\,\big(1-m(\beta,h)^2\big),
\ee
where $c_3>0$ is a certain explicit constant that we do not specify here for simplicity.
This differential inequality and the quantity $\varphi_S(\beta)$ are motivated by a similar inequality in the context of Bernoulli percolation (see \cite[(1.1) and (1.3)]{DumTas15}).

Equation \eqref{eq:diff} is proved using random currents, for which the parallel with Bernoulli percolation becomes uncanny: the role of Bernoulli percolation is replaced by $\widehat{\mathbf n_1+\mathbf n_2}$, where $\mathbf n_1$ and $\mathbf n_2$ are independent currents sampled according to an infinite-volume version of ${\bf P}^\emptyset$. Then, the proofs of \eqref{eq:diff} and its Bernoulli percolation analogue \cite[(1.3)]{DumTas15} are very close in spirit.

Inequality \eqref{eq:diff} motivates the introduction of a new critical parameter $\tilde\beta_c$ defined as the supremum of the $\beta\ge0$ for which there exists a finite set $S\ni0$ with $\varphi_\beta(S)<1$. With this definition, we automatically get that for any $\beta\ge \tilde\beta_c$,
\be\label{eq;diff}
\frac{\partial}{\partial \beta}\big(m(\beta,h)^2\big)\ge c_3\,\big(1-m(\beta,h)^2\big)
\ee
which, when integrated between $\tilde \beta_c$ and $\beta$, leads to 
$
m(\beta,h)\ge c_0(\beta-\tilde\beta_c)^{1/2}.
$
Letting $h\searrow 0$ gives us Item (1) with $\tilde\beta_c$ instead of $\beta_c$. Note that it automatically implies that $\tilde\beta_c\ge \beta_c$. 

Proving Item (3) with $\tilde \beta_c$ instead of $\beta_c$ would conclude the proof, since it would automatically imply that $\tilde\beta_c=\beta_c$. The proof of Item (3) for $\beta<\tilde\beta_c$ follows fairly quickly from the following lemma, since it implies that for any $S\ni 0$ contained in $\Lambda_n$ and any $x\in \mathbb Z^d$, $\langle \sigma_0\sigma_x\rangle_{\beta,0}\le \varphi_\beta(S)^{\|x\|/n}.$
 \begin{lemma}[Simon-Lieb inequality \cite{Lie80}]\label{lem:finiteCriterion:ising}
  Let $S$ be a finite subset of $\mathbb Z^d$ containing 0. For any $x\notin S$, 
\be\label{eq:Simon}\langle \sigma_0\sigma_x\rangle_{\beta,h}\le \sum_{y\in \partial S}\langle \sigma_0\sigma_y\rangle_{S,\beta,h}\langle \sigma_y\sigma_x\rangle_{\beta,h}.\ee
\end{lemma}

The proof of \eqref{eq:Simon} is based on the backbone representation and the interpretation in terms of random-walk attached to it. To understand intuitively \eqref{eq:Simon}, consider for a moment the simple random-walk model. Let $G(x,y)$ be the expected number of visits to $y$ starting from $x$, and $G_S(x,y)$ the same quantity when counting visits before exiting $S$. Then, the union bound and the Markov property at the first visit of $\partial S$ leads to 
$
G(0,x)\le \sum_{y\in \partial S} G_S(0,y)G(y,x),
$
which is the direct analogue of \eqref{eq:Simon}. Therefore, it does not come as a surprise that the backbone representation can be used to prove the lemma.

Item (2) with $\tilde\beta_c$ instead of $\beta_c$ is obtained via the following easy differential inequality (see \cite[(1.12)]{AizBarFer87}), which is also based on the percolation interpretation of random currents:
\be
2d \frac{\partial}{\partial h}\big(m(\beta,h)^3\big)\ge \frac{\partial}{\partial \beta}\big(m(\beta,h)^2\big).
\ee

In conclusion, the proof of Theorem~\ref{thm:1} is heavily based on both the percolation and random-walk interpretations of currents. The proof extends to any invariant ferromagnetic interactions. In fact, the differential inequality and Simon's inequality have natural analogues in this context, which are even cleaner to state, provided $\varphi_\beta(S)$ is defined slightly differently (we chose the simplest formulation here). We refer to \cite{DumTas15} for details.

\subsection{Continuity of the phase transition for the n.n.f. Ising model on $\mathbb Z^d$}\label{sec:3}

Statistical physics is often interested in the classification of infinite-volume measures of a given model. In the case of the Ising model on $\mathbb Z^d$, such measures, called {\em Gibbs measures}, are defined as probability spaces $(\{\pm1\}^{\mathbb Z^d},\mathcal F,\langle\cdot\rangle)$ satisfying the famous Dobrushin-Lanford-Ruelle condition.

We already encountered an example of Gibbs measure at inverse temperature $\beta$ in the introduction since $\langle \cdot\rangle_{\beta,0}$ is a Gibbs measure called the {\em Gibbs measure with free boundary conditions} (it is usually denoted by $\langle\cdot\rangle_\beta^0$ and we adopt this convention from now on). But one may consider the limits $\langle\cdot\rangle_\beta^+$ and $\langle \cdot\rangle_\beta^-$ of $\langle\cdot\rangle_{\beta,h}$ as $h\searrow 0$ and $h\nearrow 0$ respectively, which are also Gibbs measures (they are called {\em Gibbs measures with $+$ and $-$ boundary conditions}).

There may be many other Gibbs measures for a fixed $\beta\ge0$. In fact, one can show that there are multiple Gibbs measures if and only if $\langle\cdot\rangle_\beta^+\ne\langle\cdot\rangle_\beta^-$. This criterion implies that the Gibbs measure is unique (resp. non-unique) if $\beta<\beta_c$ (resp. $\beta>\beta_c$). Note that the case $\beta=\beta_c$ remains much more difficult to treat and it is a priori unclear whether there exist several Gibbs measures or not. 
For instance, the Ising model on $\mathbb Z$ with $J_{xy}=1/|x-y|^2$ is known to have several Gibbs measures at $\beta_c$, see \cite{AizChaCha88}. Yet, this is expected never to be the case for the n.n.f. model.

In dimension 2, Yang's result $m^*(\beta_c)=0$ implies that $\langle\cdot\rangle_{\beta_c}^+=\langle\cdot\rangle_{\beta_c}^-$. In dimension four and more, a similar result \cite{AizFer86} implied the uniqueness as well. The case of dimension 3 remained open for a while, mostly because the physical understanding of statistical physics in this dimension is slightly more limited. The following theorem fills this gap.

\begin{theorem}[\cite{AizDumSid15}]\label{thm:continuous}
There exists a unique Gibbs measure $\langle\cdot\rangle_{\beta_c}$ at $\beta_c$ for the n.n.f Ising model on $\mathbb Z^d$ with $d\ge3$. Furthermore, 
\be\label{eq:44}\frac{c_4}{\|x-y\|^{d-1}}\le \langle\sigma_x\sigma_y\rangle_{\beta_c}\le \frac{c_5}{\|x-y\|^{d-2}}\ee
for any $x,y\in\mathbb Z^d$, where $c_4,c_5\in(0,\infty)$ are universal constants.
\end{theorem}

Let us briefly describe the strategy of the proof. The idea is to prove that $\langle\cdot\rangle_{\beta_c}^+$ is equal to $\langle\cdot\rangle_{\beta_c}^0$. Indeed, this implies immediately that $\langle\cdot\rangle_{\beta_c}^+=\langle\cdot\rangle_{\beta_c}^-$ since the Gibbs measure with free boundary conditions is symmetric under global spin flip.

Let us start by saying that $\langle\cdot\rangle^0_{\beta_c}$ is known to satisfy \eqref{eq:44}. Indeed, the left-hand side can be proved using $\beta_c=\tilde\beta_c$ \footnote{Indeed, the definition of $\tilde\beta_c$ implies immediately that $\varphi_{\beta_c}(S)\ge1$ for any finite set $S$. Applying this observation to $S=\Lambda_n$ and using a few classical inequalities implies the result.}. The right-hand side is a consequence of the celebrated infrared bound (see \cite{Bis09} for a review) which states that for $\beta<\beta_c$ and $x,y\in \mathbb Z^d$,
\be\label{eq:IB}\langle\sigma_x\sigma_y\rangle_{\beta}^0\le \tfrac{1}{2\beta}\,G(x,y)\le \frac{c_5}{\|x-y\|^{d-2}},\ee
where $G(x,y)$ is the Green function of the simple random walk on $\mathbb Z^d$ (the backbone expansion gives credence for such a bound, even though the proof does not rely on it). Then, the second inequality of \eqref{eq:44} follows by taking the limit $\beta\nearrow\beta_c$ (this is possible since $\langle\cdot\rangle_{\beta}^0$ converges weakly to $\langle\cdot\rangle_{\beta_c}^0$).

The proof that $\langle\cdot\rangle_{\beta_c}^+$ is equal to $\langle\cdot\rangle_{\beta_c}^0$ is based on the study of the percolation properties of the infinite-volume limit of random duplicated currents. Very roughly, the idea is to show that the random subgraph of $\mathbb Z^d$ obtained by taking the limit as $\Lambda\nearrow \mathbb Z^d$ and then $h\searrow 0$ of the random variable $\widehat{\mathbf n_1+\mathbf n_2}$, where $\mathbf n_1$ and $\mathbf n_2$ are two independent random currents with law $\mathbf P^\emptyset_{\Lambda,\beta,h}$ and $\mathbf P^\emptyset_{\Lambda,\beta,0}$ respectively (note that the first current has a magnetic field and not the second), does not contain an infinite connected component almost surely. 

This question is reminiscent of a classical conjecture in percolation theory, namely that Bernoulli percolation on $\mathbb Z^d$ does not percolate at criticality. It is therefore a priori very difficult to prove such a statement. Nevertheless,
in our context, the fact that \eqref{eq:44} is available for $\langle \cdot\rangle_{\beta_c}^0$ can be combined with ergodic properties of the random subgraph (namely that, when it exists, the infinite connected component is unique almost surely) to prove that the random graph cannot contain an infinite-connected component almost surely. Once again, this proof combines the random-walk {\em and} the percolation perspectives.

\subsection{Truncated four-point function}\label{sec:4} To simplify the notation, we drop the dependency in $\beta$ and $h$. The random currents representation was initially introduced in \cite{Aiz82} to study (among other things) the Ursell four-point function defined for any $x_1,x_2,x_3,x_4$ as
\begin{align*}
&U_4(x_1,x_2,x_3,x_4)\\
&\ \ \ \ =\langle\sigma_{x_1}\sigma_{x_2}\sigma_{x_3}\sigma_{x_4}\rangle-\langle\sigma_{x_1}\sigma_{x_2}\rangle\langle\sigma_{x_3}\sigma_{x_4}\rangle-\langle\sigma_{x_1}\sigma_{x_3}\rangle\langle\sigma_{x_2}\sigma_{x_4}\rangle-\langle\sigma_{x_1}\sigma_{x_4}\rangle\langle\sigma_{x_2}\sigma_{x_3}\rangle.
\end{align*}
Indeed, the switching lemma enables us to rewrite $U_4(x_1,x_2,x_3,x_4)$ as
\be \label{eq:Ursell}
U_4(x_1,x_2,x_3,x_4)=-2\ \langle \sigma_{x_1}\sigma_{x_3}\rangle\langle \sigma_{x_2}\sigma_{x_4}\rangle{\bf P}^{\{x_1,x_3\}}\otimes{\bf P}^{\{x_2,x_4\}}[x_1\lr[\widehat{{\mathbf n}_1+{\mathbf n} _2}]x_2,x_3,x_4],\ee
where ${\bf P}^{\{a,b\}}$ denotes the law introduced in Definition~\ref{def:probability}.
Thus, the connectivity properties of the sum of two independent currents is once again involved in the estimation of truncated spin-spin correlations. Let us mention that random currents were used to show that Ursell $2n$-point functions is positive if and only if $n$ is odd, see \cite{Shl86}.

\medbreak\noindent
{\em 4.3.1. Triviality in dimension $d\ge5$.} 
In this section, we work with $\langle \cdot\rangle_{\beta}^0$ and n.n.f. interactions.  Since Wick's rule is equivalent to the fact that $U_4(x_1,x_2,x_3,x_4)$ vanishes, $U_4(x_1,x_2,x_3,x_4)$ is a measure of how non Gaussian the field $(\sigma_x:x\in\mathbb Z^d)$ is. More precisely, define the {\em renormalized coupling constant} 
\be
g(\beta):=\sum_{x_2,x_3,x_4\in\mathbb Z^d}\frac{U_4(0,x_2,x_3,x_4)}{\chi(\beta)^2\xi(\beta)^d},
\ee
where\footnote{The quantities $\chi(\beta)$ and $\xi(\beta)$ are well defined thanks to Theorem~\ref{thm:1} (plus an additional sub-additivity argument for the definition of $\xi(\beta)$). }
\be
\chi(\beta):=\sum_{x\in\mathbb Z^d}\langle\sigma_0\sigma_x\rangle_{\beta}^0\quad\text{ and }\quad\xi(\beta):=\left(\lim_{n\rightarrow\infty}-\tfrac1n\log\big(\langle\sigma_0\sigma_{ne_1}\rangle_\beta^0\big)\right)^{-1}
\ee
($e_1$ is a unit vector in $\mathbb Z^d$). If $g(\beta)$ tends to 0 as $\beta\nearrow\beta_c$, the field is said to be {\em trivial}. Otherwise, it is said to be {\em non-trivial}. 
Aizenman and Fr\"ohlich proved the following theorem\footnote{This theorem shed a new light on constructive quantum field theory, since it implied that the quantum field constructed from the Ising model (or more generally the $\phi^4_d$ lattice model) is simply the Free Field in dimension 5 and higher. While not fully proved yet, the same should be true in 4d, and therefore this field is the wrong candidate for a non-trivial field in 4d.}. 

\begin{theorem}[\cite{Aiz82,Fro82}]\label{thm:triviality}
For $d\ge5$, $g(\beta)$ tends to 0 as $\beta\nearrow\beta_c$.
\end{theorem}

Let us briefly discuss Aizenman's proof\,\footnote{Fr\"ohlich used the Symanzik-Brydges-Fr\"ohlich-Spencer representation to prove his theorem, with an integration by part formula replacing the switching lemma.}, which illustrates again the power of combining random-walk and percolation interpretations.
Theorem~\ref{thm:triviality} follows from the fact that the probability on the right-hand side of \eqref{eq:Ursell} tends to zero when $x_1,\dots,x_4$ are far away from each others (formulated differently, Wick's rule is asymptotically true in dimension five and higher). To prove this statement, Aizenman used the intuition coming from random-walks. In dimension five and higher, the connected components of $x_1$ and $x_2$ in $\widehat{{\mathbf n}_1+{\mathbf n} _2}$ should not be very different from the backbones of $\mathbf n_1$ and $\mathbf n_2$ respectively. Furthermore, since these backbones look like walks, they are expected to behave like simple random-walks in dimension four and higher. Therefore, the event under consideration in \eqref{eq:Ursell} should intuitively have a probability comparable to the probability that two independent simple random-walks from $x_1$ to $x_3$ and from $x_2$ to $x_4$ intersect each others. A short computation shows that they do it with probability tending to zero as soon as $d\ge4$. Of course, the connected components of $x_1$ and $x_2$ are not completely equivalent to two backbones, which themselves are not completely equivalent to two independent simple random-walks. Therefore, one needs some additional work to complete the proof. Let us simply say that the main tool in the proof is the infrared bound \eqref{eq:IB} discussed in the previous section. 

\medbreak\noindent
{\em 4.3.2. The two-dimensional case.} Equation \eqref{eq:Ursell} has a beautiful interpretation when working with spin-spin correlations on the boundary of a two-dimensional ``simply connected'' graph. 

\begin{corollary}\label{cor:2d}
Let $\beta>0$ and $\Lambda$ be a connected subgraph of $\mathbb Z^2$ with connected complement. In the formula below, $\langle\cdot\rangle$ denotes $\langle \cdot\rangle_{\Lambda,\beta,0}$. Let $x_1$, $x_2$, $x_3$ and $x_4$ be four vertices on the boundary $\partial\Lambda$ of $\Lambda$ found in counter-clockwise order (when going around the boundary), then 
\begin{align*}&\langle \sigma_{x_1}\sigma_{x_2}\sigma_{x_3}\sigma_{x_4}\rangle=\langle \sigma_{x_1}\sigma_{x_2}\rangle\langle \sigma_{x_3}\sigma_{x_4}\rangle-\langle \sigma_{x_1}\sigma_{x_3}\rangle\langle \sigma_{x_2}\sigma_{x_4}\rangle+\langle \sigma_{x_1}\sigma_{x_4}\rangle\langle \sigma_{x_2}\sigma_{x_3}\rangle.\end{align*}
\end{corollary}

The formula on the right differs from Wick's rule by a minus sign. The proof follows from the fact that the probability on the right-hand side of \eqref{eq:Ursell} is equal to 1. Indeed, the trace $\widehat{\mathbf n}_1$ of the current $\mathbf n_1$ with sources at $x_1$ and $x_3$ contains a path from $x_1$ to $x_3$, which must intersect the path from $x_2$ to $x_4$ present in $\widehat{\mathbf n}_2$ (since $\mathbf n_2$ has sources at $x_2$ and $x_4$). Therefore, $x_1$, $x_2$, $x_3$ and $x_4$ must necessarily be all connected together in $\widehat{{\mathbf n}_1+{\mathbf n} _2}$. 

The previous result extends to n.n.f. Ising models on any planar graph, and even to an arbitrary number of vertices $x_1,\dots,x_{2n}$. In such case, we obtain a {\em fermionic Wick rule} for the $2n$-point function
$$\langle \sigma_{x_1}\cdots\,\sigma_{x_{2n}}\rangle=\sum_{\pi\in\Pi_n}\varepsilon(\pi)\langle \sigma_{x_{\pi(1)}}\sigma_{x_{\pi(2)}}\rangle\dots \langle \sigma_{x_{\pi(2n-1)}}\sigma_{x_{\pi(2n)}}\rangle,$$
where $\Pi_n$ is the set of pairings of $\{1,\dots,2n\}$, i.e.~the set of permutations $\pi$ such that
$\pi(2j-1)<\pi(2j)$ for any $j\in\{1,\dots,n\}$ and $\pi(2j-1)<\pi(2j+1)$ for any $j\in\{1,\dots,n-1\}$. Above, $\varepsilon(\pi)$ is the signature of $\pi$, which can be seen as $-1$ to the power the number of intersections of the graph obtained by drawing simple arcs in $\Lambda$ between $\pi(2j-1)$ and $\pi(2j)$ for any $j\in\{1,\dots,n\}$. 

The fermionic Wick rule emerges naturally from formulae involving pfaffians (or equivalently ``Gaussian'' Grassmann integrals). Formulae expressing the partition function of the Ising model in terms of pfaffians go back to \cite{Kas63,HurGre60}. Since then, most solutions of the 2D Ising model naturally led to pfaffians formulae. Let us mention 
a direct mapping between the n.n.f. 2D Ising model and fermionic systems discovered in \cite{SchMatLie64}. In this paper, Schultz, Mattis and Lieb proved that the transfer matrix of the Ising model can be rewritten as the exponential of a quantum Hamiltonian describing a 1D chain of non-interacting fermions. 

\section{Open questions on random currents}

We now list some open questions directly related to random currents.

First, 
Equation~\eqref{eq:spin to RC} shows that the long-range order in the Ising model gets rephrased into long-range connectivity in the sum of two sourceless currents $\mathbf n_1$ and $\mathbf n_2$. One can easily check that the infinite-volume version of $\widehat{\mathbf n_1+\mathbf n_2}$ has an infinite connected component almost surely if and only if $\beta>\beta_c$. 
\medbreak
\noindent{\bf Question 1.} Does the infinite-volume version of {\em one} sourceless current $\widehat{\mathbf n}$ contain an infinite connected component almost surely at $\beta>\beta_c$?
\medbreak
Many of the correlations inequalities available for the Ising model can be obtained via random currents with the notable exception of the famous FKG inequality. One of the reasons for this failure is that random currents do not seem to behave well regarding the natural ordering on subsets of $\mathcal E$ given by the inclusion. On the contrary, the random-cluster model is naturally ordered, in the sense that there exists a coupling between $E\sim \phi_{\Lambda,\beta,h}$ and $E'\sim\phi_{\Lambda,\beta',h}$ such that $E\subset E'$ almost surely as soon as $\beta\le \beta'$. One may convince oneself that such a coupling does not exist for random currents. Nevertheless, some properties of random currents should still be increasing in $\beta$ (e.g.~${\bf P}^\emptyset_\beta\otimes{\bf P}^\emptyset_\beta[x\lr[\widehat{{\mathbf n}_1+{\mathbf n} _2}]y]=\langle\sigma_x\sigma_y\rangle_\beta^2$). As a toy example, we propose the following question.
\medbreak
\noindent{\bf Question 2.} Is $\beta\longmapsto{\bf P}^\emptyset_\beta\otimes{\bf P}^\emptyset_\beta[A\lr[\widehat{{\mathbf n}_1+{\mathbf n} _2}]B]$ increasing for any $A,B\subset \Lambda$?
\medbreak
Remark~\ref{rmk:comparison} relates the high-temperature expansion and the random-cluster model to random currents. 
 On the square lattice, both the high-temperature expansion (which corresponds to the low-temperature expansion on the dual lattice by Kramers-Wannier duality) and the interfaces of the random-cluster model were proved to be conformally invariant \cite{CheDumHon14} (the preprints \cite{BenHon16,KemSmi15} prove convergence to CLE(3) and CLE(16/3) respectively). See also \cite{DumSmi12} for a review referencing the previous contributions. \medbreak\noindent
{\bf Question~3.}~Prove~that~the~scaling~limit~of~2D~sourceless~random~currents~is~CLE(3).
\medbreak
Random currents have been geared to study truncated correlations at $h=0$ (for $h\ne0$, several arguments show that they decay exponentially fast at any $\beta>0$). As seen above, they enabled to prove that truncated spin-spin correlations decay exponentially fast for $\beta<\beta_c$  and algebraically fast for $\beta=\beta_c$ (truncating is not necessary in these cases). In 2D, Kramers-Wannier duality together with the exponential decay for $\beta<\beta_c$ imply that truncated two-point functions decay exponentially fast for $\beta>\beta_c$. The only case left is the case $\beta>\beta_c$ and $d\ge3$.
\medbreak
\noindent{\bf Question 4.} Prove that for $d\ge3$ and $\beta>\beta_c$, there exists $c_6=c_6(\beta)>0$ such that for any $x,y\in\mathbb Z^d$,
$
\langle \sigma_x\sigma_y\rangle_\beta^+-\langle \sigma_x\rangle_\beta^+\langle\sigma_y\rangle^+_\beta\le \exp[-c_6 \|x-y\|].
$
\medbreak
The result of Corollary~\ref{cor:2d} was known for a long time. Nevertheless, the strategy of the proof using random currents is of great interest since it generalizes to finite-range interactions \cite{AizDumWar16}.\ We believe that random currents can improve the understanding of universality for the 2d Ising model, in particular for arbitrary finite-range interactions\footnote{For instance, in the non-planar setting random currents were analyzed using lace expansion techniques to study general finite-range ferromagnetic Ising models in large dimension, see  \cite{Sak07}.} (some universality results were already obtained for finite-range interactions that are perturbations of the n.n.f. model in \cite{GiuGreMas12}). Nonetheless, studying the geometric properties of random currents at criticality is a very difficult challenge, as illustrated by the fact that we are currently unable to prove a RSW type result similar to the random-cluster version obtained in \cite{DumHonNol11}. 
\medbreak
\noindent{\bf Question 5.} Use random currents to study critical finite-range Ising models defined on $\mathbb Z^2$.
\medbreak
The question of the triviality of Ising on $\mathbb Z^4$ is still open. This question can be attacked with random currents. Similarly, it would be interesting to prove that the renormalized coupling constant does not tend to 0 on $\mathbb Z^3$.
\medbreak
\noindent{\bf Question 6.} Prove that the renormalized coupling constant converges (respectively does not converge) to 0 in dimension 4 (respectively 3).
\medbreak
Despite the fact that triviality is not proved in dimension 4, critical exponents are known to take their mean-field bound in dimension 4 (see \cite{AizFer86,AizFer88}). It would be interesting to prove that this is not the case in dimension 3. 
\medbreak
\noindent{\bf Question 7.} Consider the n.n.f. Ising model on $\mathbb Z^3$. Prove that there exists $\varepsilon,c_7,c_8>0$ such that for any $x,y\in\mathbb Z^d$,
\be\frac{c_7}{\|x-y\|^{2-\varepsilon}}\le \langle\sigma_x\sigma_y\rangle_{\beta_c}\le \frac{c_8}{\|x-y\|^{1+\varepsilon}}.\ee
\medbreak
Let us finish by a question for percolation aficionados (answering this question would provide a direct proof that 
$\langle\cdot\rangle_\beta^0=\tfrac12\langle\cdot\rangle_\beta^++\tfrac12\langle\cdot\rangle_\beta^-$ for any $\beta\ge0$, see \cite{AizDumSid15}). 
\medbreak\noindent
{\bf Question 8.} For\,$\beta>0$,\,show\,that\,the\,infinite\,connected component (if it exists) of the percolation model (built from infinite-volume duplicated currents) mentioned in the antepenultimate paragraph of Section~\ref{sec:3} has one end almost surely.

\paragraph{Acknowledgments} The author was funded by the NCCR SwissMap and the Swiss FNS. The author thanks D. Chelkak, D. Cimasoni, M. Harel, A. Raoufi, V. Tassion and Y. Velenik for their comments on the manuscript.

$ $\\
\textsc{hugo.duminil@unige.ch\\
D\'epartement de Math\'ematiques, Universit\'e de Gen\`eve,\\
2-4 rue du Li\`evre, 1211 Gen\`eve, Switzerland}


\begin{thebibliography}{7}

\bibitem{AizBarFer87}
M.~Aizenman, D.~J. Barsky, and R.~Fern{{\'a}}ndez, \emph{The phase
  \mbox{transition} in a general class of {I}sing-type models is sharp}, J.
  Stat. Phys. \textbf{47} (1987), no.~3-4, 343--374. 

\bibitem{AizChaCha88}
M.~Aizenman, J.~T. Chayes, L.~Chayes, and C.~M. Newman, \emph{Discontinuity of
  the magnetization in one-dimensional {$1/\vert x-y\vert ^2$} {I}sing and
  {P}otts models}, J. Stat. Phys. \textbf{50} (1988), no.~1-2, 1--40.

\bibitem{AizDumSid15}
M.~Aizenman,\,H.~Duminil-Copin,\,and\,V.~Sidoravicius, \emph{Random {C}urrents
  and \mbox{{C}ontinuity} of {I}sing {M}odel's {S}pontaneous {M}agnetization},
  Comm. Math. Phys. \textbf{334} (2015), 719--742.
  
  \bibitem{AizDumWar16}
M.~Aizenman, H.~Duminil-Copin, V. Tassion and S.~Warzel, \emph{Emerging planarity in the 2D Ising model},
 in preparation (2016).


\bibitem{AizFer86}
M.~Aizenman, and R.~Fern{\'a}ndez, \emph{On the critical behavior of the
  magnetization in high-dimensional {I}sing models}, J. Stat. Phys. \textbf{44}
  (1986), no.~3-4, 393--454.
  
\bibitem{AizFer88}
M.~Aizenman, and Roberto Fern{{\'a}}ndez, \emph{Critical exponents for
  long-range interactions}, Lett. Math. Phys. \textbf{16} (1988), no.~1,
  39--49.
  
\bibitem{Aiz82}
M.~Aizenman, \emph{Geometric analysis of {$\varphi ^{4}$} fields and {I}sing
  models. {I}, {II}}, Comm. Math. Phys. \textbf{86} (1982), no.~1, 1--48.

\bibitem{Bax89}
R.~J. Baxter, \emph{Exactly solved models in statistical mechanics},
  Academic Press Inc. [Harcourt Brace Jovanovich Publishers], London, 1989,
  Reprint of the 1982 original. 
  
\bibitem{BaxEnt78}
R. J. Baxter, and I. G. Enting, {\em 399th solution of the Ising model}, J. Phys. A: Math. and General, {\bf 11},
 (1978), no. 12, 2463.
 

\bibitem{BenHon16}
S.~Benoist, and C.~Hongler, \emph{The scaling limit of critical Ising interfaces is CLE(3)}, arXiv:1604.06975 (2016).

\bibitem{BryFroSpe82}
D. Brydges, J. Fr{{\"o}}hlich, and T. Spencer, \emph{The random
  walk representation of classical spin systems and correlation inequalities},
  Comm. Math. Phys. \textbf{83} (1982), no.~1, 123--150. 

\bibitem{Bis09}
M. Biskup, \emph{Reflection positivity and phase transitions in lattice spin
  models}, Methods of contemporary mathematical statistical physics, Lecture
  Notes in Math., vol. 1970, Springer, Berlin, 2009, pp.~1--86.


\bibitem{CheDumHon14}
D.~Chelkak, H.~Duminil-Copin, C.~Hongler, A.~Kemppainen, and S.~Smirnov,
  \emph{\mbox{Convergence} of {I}sing interfaces to {S}chramm's {SLE} curves}, C. R.
  Acad. Sci. Paris Math. \textbf{352} (2014), no.~2, 157--161.

%
%




\bibitem{DumHonNol11}
H.~Duminil-Copin, C.~Hongler, and P.~Nolin, \emph{Connection probabilities
  and {RSW}-type bounds for the two-dimensional {Fortuin-Kasteleyn} {I}sing model}, Comm. Pure
  Appl. Math. \textbf{64} (2011), no.~9, 1165--1198. 

\bibitem{DumSmi12}
H.~Duminil-Copin, and S.~Smirnov, \emph{Conformal invariance of lattice
  models}, \mbox{Probability} and statistical physics in two and more dimensions, Clay
  Math. Proc., vol.~15, Amer. Math. Soc., Providence, RI, 2012, pp.~213--276.

\bibitem{DumTas15}
H.~Duminil-Copin, and V.~Tassion, \emph{A new proof of the sharpness of the
  phase \mbox{transition} for {B}ernoulli percolation and the {I}sing model},
  Comm. Math. Phys., {\bf 343} (2016), no.~2, 725--745.

\bibitem{Dum13}
H.~Duminil-Copin, \emph{Parafermionic observables and their applications to planar
  \mbox{statistical} physics models}, Ensaios Matematicos, vol.~25, Brazilian
  Mathematical Society, 2013.

\bibitem{Dum15}
H.~Duminil-Copin, \emph{Geometric representations of lattice spin models}, book, Edition
  Spartacus, 2015.

\bibitem{ForKas72}
C.~M. Fortuin, and P.~W. Kasteleyn, \emph{On the random-cluster model. {I}.
  {I}ntroduction and relation to other models}, Physica \textbf{57} (1972),
  536--564.
  
\bibitem{Fro82}
J. Fr{\"o}hlich, \emph{On the triviality of $\lambda \varphi^{4}_{d}$
  theories and the approach to the critical point in $d\ge4$ dimensions},
  Nuclear Phys. B \textbf{200} (1982), no.~2, 281--296. 

\bibitem{FroSimSpe76}
J.~Fr{{\"o}}hlich, B.~Simon, and Thomas Spencer, \emph{Infrared bounds, phase
  transitions and continuous symmetry breaking}, Comm. Math. Phys. \textbf{50}
  (1976), no.~1, 79--95. 

\bibitem{GiuGreMas12}
A. Giuliani, R. Greenblatt, and V. Mastropietro, {\em The scaling limit of the energy correlations in non-integrable Ising models}, {J. Math. Phys.},
{\bf 53} (2012), no. 9, 095214.

\bibitem{Gri67}
R.~Griffiths, {\em Correlation in {I}sing ferromagnets {I}, {II}}, J. Math. Phys., Vol.~8, pp.~478--489, 1967.

\bibitem{GriHurShe70}
Robert~B. Griffiths, C.~A. Hurst, and S.~Sherman, \emph{Concavity of
  magnetization of an {I}sing ferromagnet in a positive external field}, J.
  Math. Phys. \textbf{11} (1970), 790--795.
  
 \bibitem{GriJan09} G. Grimmett, and S. Janson, {\em Random even graphs},
{Elec.\,J.\,Comb.},
{\bf 16R46} (2009) 1.

\bibitem{HurGre60} C.A. Hurst, H.S. Green, {\em New solution of the Ising problem for a rectangular lattice}, J. Chem. Phys,
{\bf 33}
(1960), no.~4,
1059--1062.

  


\bibitem{Isi25}
E.~Ising, \emph{Beitrag zur {T}heorie des {F}erromagnetismus.}, Z. Phys.
  \textbf{31} (1925), 253--258.


\bibitem{KacWar52}
M.~Kac and J.~C. Ward, \emph{A combinatorial solution of the two-dimensional
  {I}sing model}, Phys. Rev \textbf{88} (1952), 1332--1337.
  
  \bibitem{Kas63}
P.~W.~Kasteleyn, {\em Dimer statistics and phase transitions}, J. Math. Phys., {\bf 4} (1963) 287--293.

\bibitem{KraWan41}
H.~A. Kramers and G.~H. Wannier, \emph{Statistics of the two-dimensional
  ferromagnet. {I}}, Phys. Rev. (2) \textbf{60} (1941), 252--262.

\bibitem{Kau49}
B. Kaufman,
{\em Crystal statistics. II. Partition function evaluated by spinor analysis},
Phys. Rev.
{\bf 76} (1949), 1232--1243.

\bibitem{KemSmi15}
A. Kemppainen, and S. Smirnov, {\em 
    Conformal invariance of boundary touching loops of FK Ising model},  arXiv:1509.08858 (2015).
      
\bibitem{Leb74}
J. Lebowitz, {\em G{HS} and other inequalities},  Comm. Math. Phys. \textbf{35}
  (1974), 87--92.

\bibitem{Len20}
W.~Lenz, \emph{Beitrag zum {V}erst\"andnis der magnetischen {E}igenschaften in
  festen {K}\"orpern.}, Phys. Zeitschr. \textbf{21} (1920), 613--615.

\bibitem{Lie80}
E. H. Lieb, {\em A refinement of Simon's correlation inequality}, Comm. Math. Phys.
 {\bf 77} (1980),
 127--135.


\bibitem{LupWer15}
T. Lupu, and W. Werner, {\em A note on Ising random currents, Ising-FK, loop-soups and the Gaussian free field},
 {arXiv:1511.05524} (2015).
 

\bibitem{Ons44}
L.~Onsager, \emph{Crystal statistics. {I}. {A} two-dimensional model with an
  order-disorder transition}, Phys. Rev. (2) \textbf{65} (1944), 117--149.

\bibitem{Pei36}
R.~Peierls, \emph{On {I}sing's model of ferromagnetism}, Math. Proc. Camb.
  Phil. Soc. \textbf{32} (1936), 477--481.
  
\bibitem{Sak07}
A. Sakai, {\em Lace Expansion for the Ising Model}, 
Comm. Math. Phys., {\bf 272} (2007), no.~2, 283--344.

\bibitem{SchMatLie64}
T.D. Schultz, D. C. Mattis, and E.H. Lieb, {\em Two-Dimensional Ising Model as a Soluble Problem of Many Fermions}, Rev. Mod. Phys., {\bf 36} (1964), no. 3, 856--871.

\bibitem{Shl86} 
S. Shlosman, {\em Signs of the Ising model Ursell functions}, Comm. Math. Phys.,
{\bf 4} (1986), no.~102, 679--686.

%


\bibitem{Sym66}
K.~Symanzik, \emph{Euclidean quantum field theory. {I}. {E}quations for a
  scalar model}, J. Math. Phys. \textbf{7} (1966), 510--525.

\bibitem{vdW}
B.~L. van~der Waerden.
\newblock Die lange {R}eichweite der regelmassigen {A}tomanordnung in
  {M}ischkristallen.
\newblock {\em Z. Physik}, {\bf 118} (1941), 473--488.

\bibitem{Yan52}
C.N. Yang, \emph{The spontaneous magnetization of a two-dimensional {I}sing
  model}, Phys. Rev. \textbf{85} (1952), no.~5, 808--816.
\end{thebibliography}
\end{document}